\documentclass[sigplan,twocolumn]{acmart}

\usepackage{algorithm}
\usepackage{algpseudocode}
\usepackage{tikz}
\usetikzlibrary{arrows.meta, positioning, shapes, calc}
\usepackage[normalem]{ulem}
\usepackage{cancel}

\AtBeginDocument{%
  }

\setcopyright{acmlicensed}
\copyrightyear{2026}
\acmYear{2026}
\acmDOI{XXXXXXX.XXXXXXX}
\acmConference[PaPoC '26]{Make sure to enter the correct
  conference title from your rights confirmation email}{April 27--30,
  2026}{Edinburgh, UK}
\acmISBN{978-1-4503-XXXX-X/2018/06}

\begin{document}

\title{Semantic Conflict Model for Collaborative Data Structures}

\author{Georgii Semenov}
\email{georgii.v.semenov@gmail.com}
\orcid{0000-0003-4725-7666}
\affiliation{%
  \institution{ITMO University}
  \city{Saint Petersburg}
  \country{Russia}
}

\author{Vitaly Aksenov}
\email{aksenov.vitaly@gmail.com}
\orcid{0000-0001-9134-5490}
\affiliation{%
  \institution{ITMO University}
  \city{Saint Petersburg}
  \country{Russia}
}


\newcommand{\fix}{\color{red}}

\begin{abstract}

Digital collaboration systems support asynchronous work over replicated data,
where conflicts arise when concurrent operations cannot be unambiguously integrated
into a shared history. While Conflict-Free Replicated Data Types (CRDTs) ensure
convergence through built-in conflict resolution, this resolution is typically
implicit and opaque to users, whereas existing reconciliation techniques often
rely on centralized coordination.
This paper introduces a conflict model for collaborative data structures that enables explicit,
local-first conflict resolution without central coordination. The model identifies
conflicts using semantic dependencies between operations and resolves them by rebasing
conflicting operations onto a reconciling operation via a three-way merge over a replicated journal.
We demonstrate our approach on collaborative registers, including an explicit formulation
of the Last-Writer-Wins Register and a multi-register entity supporting semi-automatic reconciliation.



\end{abstract}

\ccsdesc[300]{Human-centered computing~Computer supported cooperative work}
\ccsdesc[500]{Computing methodologies~Distributed algorithms}
\ccsdesc[500]{Information systems~Remote replication}

\keywords{semantic conflicts, optimistic concurrency control,
eventual consistency, collaborative editing, CRDT}

\received{xx February 2026}
\received[revised]{xx March 2026}
\received[accepted]{xx April 2026}

\maketitle

\tikzset{
  obj/.style={circle, draw, inner sep=1pt},
  base/.style={regular polygon, regular polygon sides=3, draw, inner sep=2pt},
  entails/.style={draw, -{Stealth}, thick},
  discards/.style={draw, -{Stealth}, thick},
  conflict/.style={draw, thick},
  cross/.style={node contents={$\,\times\,$}}
}

\section{Introduction}

Digital collaboration environments provide a shared workspace that enables multiple users,
teams, or organizations to work together toward a common goal \cite{Schmidt1992}.
Implementing highly available groupware systems, however, is challenging due to the need to 
tolerate network partitions and enable isolated work, as is common in version concurrency control systems.
To address these challenges, users are allowed to collaborate asynchronously by producing
{\itshape optimistic} transactions that are applied
locally to their data replicas, and eventually propagated to other replicas \cite{SaitoShapiro2005}.
This approach has been revisited as {\itshape eventual state-machine replication} \cite{Kuznetsov2025}.

The core part of such systems is the synchronization algorithm, which
ensures convergence to a single state once all contributions have been integrated.
During synchronization, {\itshape conflicts} arise when concurrent operations, i.e., that
are not causally related, interfere with one another. This may be captured semantically, for example, 
as operations violating each other's intents \cite{EllisGibbs1989, Sun1998} 
or as operations being non-commutative \cite{Weihl1988}.
The common emergence of conflicts, however, is in the ambiguity of how operations should be integrated into a shared history.
Crucially, the multi-dimensionality of this ambiguity elevates conflict resolution to a distinct collaborative subtask, concerned with coordinating and
reconciling parallel commitments, a process referred to as {\itshape articulation work} \cite{Schmidt1992}.

One widely adopted approach to conflict resolution is exemplified by
conflict-free replicated data types (CRDTs) \cite{Shapiro2011},
which rely on semilattice structures or on the commutativity of operations.
CRDTs provide strong eventual consistency through resolution strategies that are embedded directly into the data
type definition, as captured in Add-Wins Sets and Last-Writer-Wins Registers \cite{Shapiro2011}.
Multi-Value Registers \cite{Shapiro2011} take a different stance by retaining all concurrently written values and thereby deferring their resolution.
Nevertheless, conflict resolution in CRDTs generally remains implicit, opaque to users, cannot be supervised,
and is generally non-native to application-specific semantics.

A second, more expressive class of approaches is based on reconciliation techniques \cite{ShapiroRowstronKermarrec2000}, 
which treat conflict resolution as an application-level
procedure that integrates concurrent transactions in a semi-automatic manner.
Typical strategies include computing an execution order that minimizes violated constraints
\cite{PreguicaShapiroMatheson2003} or maximizes the number of preserved operations \cite{Semenov2007}.
These techniques usually rely on a centralized reconciliation phase with global knowledge of
all operations, which makes them poorly suited to local-first \cite{Kleppmann2019:LocalFirst} collaboration
scenarios, where replicas must continue to make progress independently on local copies.
This limitation motivates the need for reconciliation middleware that operates without a central coordinator,
provides explicit convergence guarantees, and treats conflict resolution as a first-class collaborative activity,
that is, as a form of articulation work.

This paper presents a novel conflict model for collaborative data structures
that enables local-first conflict resolution through the reconciliation of operations.
Conflicts are identified using {\itshape entails} and {\itshape discards} dependencies,
which relate each operation to its logical premises, and
conflicting operations are rebased onto a new reconciling operation.
We illustrate the model through several examples of collaborative registers and their formal definitions,
including the Last-Writer-Wins Register as an instance of a CRDT with automated conflict resolution,
and a calendar event as an example of a multi-register entity supporting complex operations.

\section{System Model}

We consider a shared memory $M$ which is modeled as a set
of $n$ pre-initialized registers. Each register $M_i$ has its basis of {\itshape actions} $A_i$,
which is formally a set of all allowed atomic operations on the register. We define an {\itshape operation} $o = [a_1, a_2, \ldots, a_k]$
as a contiguous sequence of action instances $a_i \in \bigcup^n A_i$ on the registers with the transactional semantics,
i.e., executing an operation $o$ is equivalent to a sequential uninterruptible execution of all its actions $a_1 \circ a_2 \circ \ldots \circ a_k$.
The {\itshape state} $s_i \in S_i$ of register $M_i$ is deterministically produced by an interpreter $\mathcal{I}_i$ after executing
a sequence of actions on the initial empty context $\perp$.
Equality of registers is inferred from the equality of their states, i.e.,
$x = y \Leftrightarrow s_x = s_y$.

Our distributed system model consists of multiple nodes with 
their shared memory replicas $\mathcal{R} = \{r_1, r_2, \ldots, r_n\}$.
Operations issued on the replicas are {\itshape applied} locally and eventually {\itshape synchronized}
between nodes, which are connected with an unreliable network.
Local replica $r$ of the shared memory is modeled
by a {\itshape local history} $\mathcal{H}_{r}$, which is a sequence of operations, and each register $x = M_i$ has a projection $\mathcal{H}_{r}(x)$
of the history, which is a sequence of actions corresponding to $x$.
One can provide the notion of {\itshape happens-before} \cite{Lamport1978} relation $\prec$, defined over actions and operations in the history, so that if $o_1 \prec o_2$, then
$o_1$ precedes $o_2$ in any history $\mathcal{H}_{r}$ (same definition for $a_1 \prec a_2$ and $\mathcal{H}_{r}(x)$). 

Each memory register $x = M_i$ is initialized with a {\itshape constructor} operation $c_x$, so that $\forall a \in \mathcal{H}_{r}(x) \colon a \not\in c_x \Rightarrow c_x \prec a$.
When an operation is {\itshape applied} by a replica, it is just appended to the local history $\mathcal{H}_{r}$.
Eventually, the current replica history $\mathcal{H}_{r}$ is {\itshape published}, so that it is available for pull on the other replicas.
When ready to reconcile with foreign changes, replica $r$ {\itshape synchronizes-with} another replica $r'$ by integrating operations from another replica's 
published history $\mathcal{H}_{r'}$ into the local history, i.e., $\mathcal{H}_{r} \gets \mathcal{H}_{r} \sqcup \mathcal{H}_{r'}.$
Complete synchronization is achieved through integrating the histories of all replicas in an arbitrary order, until
they converge to the same {\itshape state} once all the conflicts are resolved.

\section{Conflict Model}

A conflict model provides mechanisms for {\itshape identification}
and {\itshape resolution} of conflicts in collaborative data structures. The strongest
conflict model is the {\itshape concurrent} conflict model, i.e.,
operations are conflict-prone if they are concurrent in terms of happens-before relation \cite{Lamport1978}:
$o_1 \parallel o_2 \Leftrightarrow  o_1 \not\prec o_2 \land o_2 \not\prec o_1$.
The weaker one is the {\itshape non-commutative} conflict model, i.e., operations conflict if they
are concurrent and non-commutative \cite{Shapiro2011}.
Usually, conflict resolution is captured by establishing the happens-before order between conflicting operations,
so that they can be applied in a consistent manner on all replicas.
In this section, we introduce a conflict model based on operation dependencies to empower
the synchronization procedure in the outlined system model.

\subsection{Conflict Identification}

Let us consider two kinds of relations over operations in the history:
{\itshape entails} ($\vdash$) and {\itshape is discarded by} ($\ll$).
Intuitively, entailment captures {\itshape epistemic preconditions} \cite{DynamicEpistemicLogic2008} of an operation: issuing a new operation relies on the presence
of effects of its premises, under the assumption that their effects remain valid. Discarding,
in turn, captures {\itshape epistemic revision} \cite{Baltag2008}: an operation may invalidate the effect of another
operation, thereby operation is no longer a premise for subsequent operations.

Syntactically, upon being applied, $o$ is assigned an immutable set of {\itshape premises}
$\Gamma = \{ o_1, o_2, \ldots, o_n \}$ that {\itshape entail} $o$, which is denoted as $\Gamma \vdash o$.
The premises of an operation $o = [a_1, \ldots, a_k]$ are derived from the premises of its actions,
i.e., $\Gamma(o) = \bigcup_{j \in [1, k]} \Gamma(a_j)$, and $a_1 \vdash a_2$ denotes that
the corresponding operations $a_1 \in o_1$ and $a_2 \in o_2$ satisfy $o_1 \vdash o_2$.
Common premises of several operations are denoted as $\Gamma(o_1, o_2, \cdots, o_n) = \bigcap_{i = 1}^n \Gamma(o_i)$.
Entailment implies happens-before relation, i.e., $o_1 \vdash o_2 \Rightarrow o_1 \prec o_2$.
Let $\vdash^{*}$ define the reflexive transitive closure of $\vdash$.
Then, concurrent operations may be redefined in terms of entailment as follows:
$o_1 \parallel o_2 \Leftrightarrow (o_1 \not\vdash^{*} o_2) \ \land \ (o_2 \not\vdash^{*} o_1)$.

We require that each action $a$ is associated with a monotonically non-increasing predicate $vis(a, \mathcal{H}_r(x))$,
indicating whether the effect of $a$ is visible in the current history $\mathcal{H}_{r}(x)$.
State is modeled as a set of actions with their visible effects.
Next, we allow an action in the history to be {\itshape rolled back} by removing it together with all actions that transitively depend on it,
i.e., $\mathcal{H}_{r}(x) \setminus {a} \equiv [a' \in \mathcal{H}_{r}(x) \colon a \centernot\vdash^{*} a']$.
Then, we define {\itshape is discarded by} relation as follows:
$a_1 \ll a_2 \Leftrightarrow a_1 \vdash a_2 \ \land \ vis(a_1, \mathcal{H}_r(x) \setminus \{a_2\}) \land \neg vis(a_1, \mathcal{H}_r(x))$.
Similarly, $o_1 \ll o_2 \Leftrightarrow \exists a_1 \in o_1, \ a_2 \in o_2 \colon a_1 \ll a_2$, and
$vis(o) = \bigwedge_{a \in o} vis(a)$.

A {\itshape conflict} arises when a premise of one operation is concurrently
discarded by another operation.
Formally, for a set of operations $\{ o_1, o_2, \ldots, o_n \}$, we define the set of {\itshape conflicting premises} as
$\hat{\Gamma}(o_1, o_2, \cdots, o_n) = \{ o' \in \Gamma(o_1, o_2, \cdots, o_n) \mid \ \exists o_i \neq o_j \colon o' \vdash o_i \land o' \ll o_j \} $.
Two operations $o_1$ and $o_2$ are said to be {\itshape compatible} if they do not have conflicting premises, i.e.
$o_1 \triangleleft o_2 \Leftrightarrow \hat{\Gamma}(o_1, o_2) = \varnothing$.
Similarly, a local history $\mathcal{H}_r$
is compatible with an operation $o$ if all operations in the history are compatible with $o$:
$\mathcal{H}_r \triangleleft o \Leftrightarrow \forall o' \in \mathcal{H}_r \colon o' \triangleleft o$.
A history $\mathcal{H}_{r}$ is {\itshape valid} if it is a topological ordering of the received operations
with respect to $\vdash$, and all pairs of operations in the history are mutually compatible:
$\forall o_i, o_j \in \mathcal{H}_{r} \colon o_i \triangleleft o_j$.

To enable representation of operations with their premises, we use an {\itshape entailment graph}, where
nodes are the operations of the local history $\mathcal{H}_r$ and directed edges represent the partial order $\vdash$.
An edge between $o_1$ and $o_2$ exists if and only $o_1 \vdash o_2$.
Then, topological sorting of the entailment graph {\itshape induces} a valid history $\mathcal{H}_r$ \cite{Kuznetsov2025}.

\subsection{Conflict Resolution}

To support conflict resolution, we allow an operation to be {\itshape rebased} with respect to its premises.
Then, intuitively, when a conflict occurs, conflicting operations may be rebased to a {\itshape merge} operation
that reconciles them.

Let us denote by  $\hat{o} \ \hat{\vdash} \ o$ that an operation $o$ is rebased to $\hat{o}$.
Formally, when $\hat{o} \ \hat{\vdash} \ o$, we interpret $o$ as an empty sequence of actions and
include its premises to the premises of $\hat{o}$,
i.e., $\Gamma(o) \subseteq \Gamma(\hat{o})$.
Additionally, we close the entailment relation transitively under rebasing:
$(o_1 \ \vdash \ o') \land (o' \ \hat{\vdash} \ o_2) \Rightarrow (o_1 \vdash o_2)$.
The incoming edges of $o$ may be removed from the entailment graph, as its premises $\Gamma(o)$ are transitively included in $\hat{o}$.
That is, rebasing induces causality:  $o_1 \ \hat{\vdash} \ o_2 \Rightarrow o_1 \prec o_2$.

To support cancellation of operations, we define a distinguished {\itshape tombstone} operation $o_\varnothing$.
An operation rebased to $o_\varnothing$ is treated as semantically null:
it does not contribute effects and need not be merged into replica histories.
Such rebased operations act as discontinuity markers and may be omitted from
subsequent conflict detection.

Now, we are ready to present the synchronization algorithm empowered by the rebase procedure: 
a replica $r$ incrementally integrates operations received from another replica $r'$ into its local history $\mathcal{H}_r$.

The procedure {\scshape sync}$(r')$, shown in Figure~\ref{algorithm:sync},
fetches the recent history $\mathcal{H}_{r'}$ from $r'$ and processes its operations
in causal order. For each operation $o \in \mathcal{H}_{r'}$, the algorithm proceeds as follows.
First, if $o$ has been rebased at $r'$, that is, there exists $\hat{o} \in \mathcal{H}_{r'}$
such that $\hat{o} \ \hat{\vdash} \ o$, the operation $o$ is is rolled back from the local history $\mathcal{H}_r$,
i.e., $o$ and all its transitive premises are removed.
If $o$ has been rebased to the tombstone operation $o_\varnothing$, it is treated as cancelled
and skipped. If an operation is rebased in both replicas to different operations, all the rebases are applied
(the corresponding merge operations $\{ \hat{o}_i \}$ may be stored in a Grow-only Set \cite{Shapiro2011}).
If $o$ is not already present in $\mathcal{H}_r$ and is compatible with the current history,
i.e., \ $\mathcal{H}_r \triangleleft o$, then $o$ is appended to $\mathcal{H}_r$ using
{\bfseries\scshape apply}, as the premises of $o$ are already in $\mathcal{H}_r$ due to the causal order of $\mathcal{H}_{r'}$.
Otherwise, if $o$ is not present in $\mathcal{H}_r$ and conflicts with it,
i.e.\ $\mathcal{H}_r \centernot\triangleleft o$, synchronization is suspended and
conflict resolution is triggered via {\bfseries\scshape resolve} procedure.

The resolution procedure {\bfseries\scshape resolve}, described below,
takes as input a local history $\mathcal{H}_{r}$ and a conflicting remote operation $o_0 \in \mathcal{H}_{r'}$ such that
$\mathcal{H}_{r} \centernot\triangleleft o_0$:

\begin{enumerate}
  \item Find all common premises $\Gamma = \bigcup_{i \in [0, k]} \Gamma(o_i)$,
  where $o_1, \ldots, o_k \in \mathcal{H}_r$ and $\forall i \in [1, k] \colon o_i \centernot\triangleleft o_0$.
  \item Define the set of operations for conflict resolution $O$ as operations
  entailed by $\Gamma$, i.e.,
  $O = \{x \mid \exists \bar{o} \in \Gamma \colon \bar{o} \vdash x \}$.
  \item Optionally, if a user participates in the conflict resolution and authorized to cancel operations,
  unneeded operations with their respective ancestors can be excluded before merging $O$.
  Therefore, $O$ is partitioned into a set of operations to preserve $\mathcal{D}_{\hat{o}}$
  and operations to cancel $\mathcal{D}_{\varnothing}$. By default, $\mathcal{D}_{\hat{o}} = O$, $\mathcal{D}_{\varnothing} = \varnothing$. 
  \item Create merge operation $\hat{o}$ with premises $\Gamma(\mathcal{D}_{\hat{o}})$ and a sequence of actions
  returned by an external {\scshape reconcile} procedure applied to the whole $\mathcal{D}_{\hat{o}}$.
  In general, {\scshape reconcile} is an interactive procedure where user manually defines desired order of actions in $\hat{o}$
  and, optionally, adds new actions or removes existing ones, as long as no new premise $p$,
  such that $p \centernot{\vdash^*} \Gamma(o_i)$, extends $\Gamma(\hat{o})$.
  \item Rebase operations in $\mathcal{D}_{\hat{o}}$ to $\hat{o}$ and operations in $\mathcal{D}_{\varnothing}$ to $o_\varnothing$.
\end{enumerate}

\begin{figure}
  \centering

  \begin{algorithmic}

  \State {\bfseries\scshape apply}$(o)$:
  \State $\mathcal{H}_{r} \gets \mathcal{H}_{r} \sqcup \{o\}$
  \\

  \State {\bfseries\scshape sync}$(r')$:
  \State $\mathcal{H}_{r'} \gets$ {\scshape fetch} $r'$ \Comment{fetch recent history from $r'$}
  \For{each $o \in \mathcal{H}_{r'}$} \Comment{in causal order}

    \If{$\exists \ \hat{o} \in \mathcal{H}_{r'} \colon \ \hat{o} \ \hat{\vdash} \ o$}
      \State $\mathcal{H}_{r} \gets \textsc{rollback}(o, \mathcal{H}_{r})$ \Comment{apply recent rebasing} 
      \If{$o_{\varnothing} \ \hat{\vdash}^{*} \ o$} 
        \State {\bfseries continue} \Comment{skip cancelled operation}
      \EndIf
    \EndIf

    \If{$o \not\in \mathcal{H}_{r} \land \mathcal{H}_{r} \triangleleft o$}
      \State $\textsc{apply}(o)$
    \ElsIf{$o \not\in \mathcal{H}_{r} \land \mathcal{H}_{r} \centernot\triangleleft o$}
      \State $\mathcal{H}_{r} \gets \textsc{resolve}(o, \mathcal{H}_{r})$ \Comment{resolve conflict}
    \EndIf
  \EndFor

  \end{algorithmic}

  \caption{Synchronization algorithm to integrate foreign history $\mathcal{H}_{r'}$ into $\mathcal{H}_{r}$.}
  \Description{Synchronization algorithm to integrate foreign history $\mathcal{H}_{r'}$ into $\mathcal{H}_{r}$.}
  \label{algorithm:sync}
\end{figure}

\begin{figure}
  \centering
  
    \begin{tikzpicture}[node distance=1.0cm]
      \node[obj] (o') {$o'$};

      \node[obj, below left=of o'] (o1) {$o_1$};
      \node[obj, below right=of o'] (o2) {$o_2$};

      \node[base, below=0.5cm of o1] (b1) {};
      \node[base, below=0.5cm of o2] (b2) {};
      \draw[entails] (o1) -- node[above, xshift=-8pt, yshift=-3pt] {$\vdash^*$} (b1);
      \draw[entails] (o2) -- node[above, xshift=8pt, yshift=-3pt] {$\vdash^*$} (b2);

      \draw[entails] (o') -- node[above, xshift=-8pt] {$\vdash$} (o1);
      \draw[discards] (o') -- node[above, xshift=8pt] {$\ll$} (o2);

      \node at (3, -1) {$\Rightarrow$};

      \node[obj] (o'2) at (5.5, 0) {$o'$};

      \node[obj, below left=of o'2] (o1b) {\cancel{$o_1$}};
      \node[obj, below right=of o'2] (o2b) {\cancel{$o_2$}};

      \node[base, below=0.5cm of o1b] (b1b) {};
      \node[base, below=0.5cm of o2b] (b2b) {};

      \draw[entails] (o1b) -- node[above, xshift=-8pt, yshift=-3pt] {$\vdash^*$} (b1b);
      \draw[entails] (o2b) -- node[above, xshift=8pt, yshift=-3pt] {$\vdash^*$} (b2b);

      \node[obj, below=0.35cm of o'2] (om) {$\hat{o}$};
      \draw[entails, dashed] (om) -- node[above, ] {$\hat{\vdash}$} (o1b);
      \draw[entails, dashed] (om) -- node[above,] {$\hat{\vdash}$} (o2b);
      \draw[entails] (o'2) -- node[right] {$\ll$} (om);

  \end{tikzpicture}
  \caption{Conflict resolution of $o_1$ and $o_2$ with one conflicting premise $o'$.}
  \Description{Two operations $o_1$ and $o_2$ have a common premise $o'$, which entails $o_1$ and is discarded by $o_2$.
  After the conflict is resolved, both $o_1$ and $o_2$ are rebased to a merge operation $\hat{o}$, which discards premise $o'$.}
  \label{fig:conflict-resolution-1}
\end{figure}
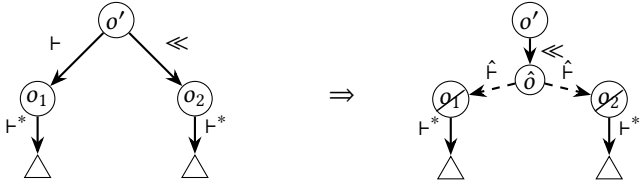

\begin{figure}
  \centering
    \begin{tikzpicture}[node distance=1.0cm]

      \node[obj] (o') {$o'$};
      \node[obj, right=of o'] (o'') {$o''$};

      \node[obj, below=of o'] (o1) {$o_1$};
      \node[obj, below=of o''] (o2) {$o_2$};
      \node[obj, right=of o2] (o3) {$o_3$};

      \node[base, below=0.5cm of o1] (b1) {};
      \node[base, below=0.5cm of o2] (b2) {};
      \node[base, below=0.5cm of o3] (b3) {};
      \draw[entails] (o1) -- node[above, xshift=-8pt, yshift=-3pt] {$\vdash^{*}$} (b1);
      \draw[entails] (o2) -- node[above, xshift=8pt, yshift=-3pt] {$\vdash^{*}$} (b2);
      \draw[entails] (o3) -- node[above, xshift=8pt, yshift=-3pt] {$\vdash^{*}$} (b3);

      \draw[entails] (o') -- node[above, xshift=-8pt] {$\vdash$} (o1);
      \draw[entails] (o'') -- node[above, xshift=8pt] {$\vdash$} (o2);
      \draw[entails] (o'') -- node[above, xshift=8pt] {$\vdash$} (o3);

      \draw[discards] (o'') -- node[above, xshift=-6pt, yshift=8pt] {$\ll$} (o1);
      \draw[discards] (o') -- node[above, xshift=+6pt, yshift=8pt] {$\ll$} (o2);

      \node at (3.68, -1) {$\Rightarrow$};

      \node[obj] (om) at (5.5, -0.95) {$\hat{o}$};
      \node[obj, above left=of om] (o'2) {$o'$};
      \node[obj, above right=of om] (o''2) {$o''$};

      \node[obj, below left=of om] (o1b) {$\cancel{o_1}$};
      \node[obj, below left=of om, xshift=1.0cm] (o2b) {$\cancel{o_2}$};
      \node[obj, below right=of om] (o3b) {$\cancel{o_3}$};

      \node[base, below=0.5cm of o1b] (b1b) {};
      \node[base, below=0.5cm of o2b] (b2b) {};
      \node[base, below=0.5cm of o3b] (b3b) {};   

      \draw[entails] (o1b) -- node[above, xshift=-8pt, yshift=-3pt] {$\vdash^*$} (b1b);
      \draw[entails] (o2b) -- node[above, xshift=-8pt, yshift=-3pt] {$\vdash^*$} (b2b);
      \draw[entails] (o3b) -- node[above, xshift=-8pt, yshift=-3pt] {$\vdash^*$} (b3b);

      \draw[entails, dashed] (om) -- node[above left,] {$\hat{\vdash}$} (o1b);
      \draw[entails, dashed] (om) -- node[above, xshift=-4pt, yshift=-6pt] {$\hat{\vdash}$} (o2b);
      \draw[entails, dashed] (om) -- node[above right,] {$\hat{\vdash}$} (o3b);
      \draw[entails] (o'2) -- node[right] {$\ll$} (om);
      \draw[entails] (o''2) -- node[left] {$\vdash$} (om);

  \end{tikzpicture}
  \caption{Conflict resolution of $o_2$ against $o_1$ and $o_3$ with two conflicting premises $o'$ and $o''$, and the merge operation $\hat{o}$ discarding premise $o'$.}
  \Description{Two operations $o_1$ and $o_2$ have common premises $o'$ and $o''$, 
  where $o'$ entails $o_1$ and is discarded by $o_2$, while $o''$ entails both $o_2$ and $o_3$. 
  $o_1$ is discarded by $o''$. After the conflict is resolved, $o_1$, $o_2$, and $o_3$ are rebased to a merge operation $\hat{o}$, which discards premise $o'$.}
  \label{fig:conflict-resolution-2}
\end{figure}
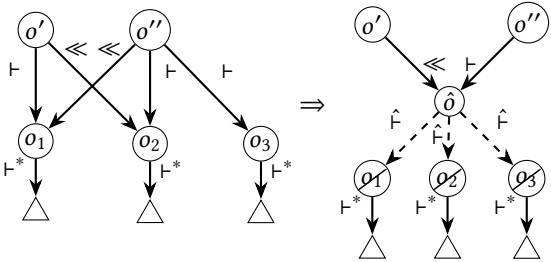

Simple cases with one and two conflicting premises
are illustrated in Figures \ref{fig:conflict-resolution-1} and \ref{fig:conflict-resolution-2} respectively.
For simplicity, common premises may be 
depicted simply as $\Gamma$. When several operations $o_1, o_2$ discard same premises $\Gamma$,
it means that $\exists o', o'' \in \Gamma \mid o' \ll o_1 \ \land \ o'' \ll o_2$.

Conflicts may be resolved concurrently, i.e., it is possible for a single operation $o$ to be rebased to different
merge operations $\hat{o}_1$ and $\hat{o}_2$ on different replicas, as illustrated in Figure \ref{fig:conflict-resolution-concurrent}.
It is an expected situation, which may lead to necessity of further conflict resolution as an inevitable part of an articulation work.

\begin{figure}
  \centering
    \begin{tikzpicture}[node distance=1.0cm]

        \node[obj] (gamma) {$\Gamma$};

        \node[obj, below left=of gamma] (o1) {$o_1$};
        \node[obj, below right=of gamma] (o2) {$o_2$};

        \node[base, below=0.5cm of o1] (b1) {};
        \node[base, below=0.5cm of o2] (b2) {};
        \draw[entails] (o1) -- node[above, xshift=-8pt, yshift=-3pt] {$\vdash^{*}$} (b1);
        \draw[entails] (o2) -- node[above, xshift=8pt, yshift=-3pt] {$\vdash^{*}$} (b2);

        \draw[discards] (gamma) -- node[left, yshift=4pt] {$\ll$} (o1);
        \draw[discards] (gamma) -- node[right, yshift=4pt] {$\ll$} (o2);

        \node at (2.0, -1) {$\Rightarrow$};  
        \node[obj] (om1) at (2.75, -0.95) {$\hat{o}_1$};
        \node[obj] (om2) at (5.25, -0.95) {$\hat{o}_2$};

        \node[obj] (o'2) at (4.0, 0.0) {$\Gamma$};

        \node[obj, below=of o'2, yshift=-0.75cm, xshift=-0.5cm] (o1b) {$\cancel{o_1}$};
        \node[obj, below=of o'2, yshift=-0.75cm, xshift=0.5cm] (o2b) {$\cancel{o_2}$};

        \node[base, below=0.5cm of o1b] (b1b) {};
        \node[base, below=0.5cm of o2b] (b2b) {};  

        \draw[entails] (o1b) -- node[above, xshift=-8pt, yshift=-6pt] {$\vdash^*$} (b1b);
        \draw[entails] (o2b) -- node[above, xshift=8pt, yshift=-6pt] {$\vdash^*$} (b2b);

        \draw[entails, dashed] (om1) -- node[above left, yshift=-8pt] {$\hat{\vdash}$} (o1b);
        \draw[entails, dashed] (om2) -- node[above right, yshift=-8pt] {$\hat{\vdash}$} (o2b);
        \draw[entails, dashed] (om2) -- node[above left, yshift=0pt] {$\hat{\vdash}$} (o1b);
        \draw[entails, dashed] (om1) -- node[above right, yshift=0pt] {$\hat{\vdash}$} (o2b);

        \draw[entails] (o'2) -- node[right] {$\ll$} (om1);
        \draw[entails] (o'2) -- node[left] {$\ll$} (om2);

  \end{tikzpicture}
  \caption{Conflict resolution of $o_1$ and $o_2$ with two conflicting premises $o'$ and $o''$ and
  a conflict resolved concurrently leading to another conflict between $\hat{o}_1$ and $\hat{o}_2$.}
  \Description{Two operations $o_1$ and $o_2$ have common premises $\Gamma$,
  where $\Gamma$ is discarded by both $o_1$ and $o_2$. After the conflict is resolved concurrently on two replicas,
  $o_1$ and $o_2$ are rebased to $\hat{o}_1$ and  $\hat{o}_2$, both discarding premise $\Gamma$ once again.}

  \label{fig:conflict-resolution-concurrent}
\end{figure}
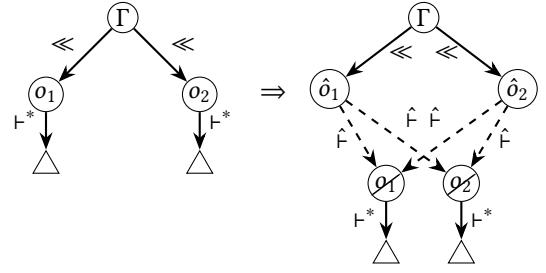

\subsection{Register Definition}

Let us discuss requirements on actions and their visibility predicates to define a register suitable for our conflict model.

\begin{definition}
  A register $x$ with an action basis $A$ is {\itshape discard-complete} if for every pair of its action instances $a_1, a_2$ in any history $\mathcal{H}_r(x)$
  of $x$ the following holds:
  $$
  vis(a_1, \mathcal{H}_r(x) \setminus \{a_2\}) \land \neg vis(a_1, \mathcal{H}_r(x))
  \Rightarrow
  a_1 \ll a_2.
  $$
\end{definition}

Then, to define a register $x = M_i$ for our shared memory model, it is sufficient to specify:
\begin{itemize}
  \item Action basis $A_i$;
  \item Interpretation function $\mathcal{I} \colon \mathcal{H}_r \to S_i$ mapping local history to the state $s \in S_i$;
  \item Arbitrary entailment rule, defining premises $\Gamma(a)$ for each action $a \in A_i$, depending on its semantics;
  \item Visibility predicate $vis(a, \mathcal{H}_r)$, defining when the effect of an action is visible, so that
  register is discard-complete.
\end{itemize}

\subsection{Correctness}

The presented shared memory model is not strongly convergent, as we expect conflicts to may have indefinitely arisen during synchronization.
However, we may prove that if no new operations are introduced and all conflicts are resolved, the shared memory strongly converges to the same state on all replicas.
We present a proof sketch which is split into two parts: firstly, we prove that the entailment graph as a graph data type satisfying {\itshape strong eventual consistency} (SEC) 
\cite{Shapiro2011SSSS} (beyond reconciliation scenarios),
and secondly, we show that each register converges to the same state in case all conflicts are resolved (assuming conflicts have been resolved in a centralized manner).

\begin{theorem}
  A graph data type, equivalent to the entailment graph $G_{\vdash}$, which supports operations $add(\Gamma, o)$ (adding a new vertex $o$ with its incoming edges $\Gamma$) 
  and $rebase(o, \hat{o})$ (adding an edge to $o$ from $\hat{o}$ not succeeding it) satisfies SEC \cite{Shapiro2011SSSS}, i.e., all replicas that have delivered 
  the same updates have the equivalent states.
\end{theorem}

\begin{proof}[Proof sketch]
  In this proof, we will rely on the definition of add-only monotonic directed acyclic graph CRDT in \cite{Shapiro2011},
  which models partial order graph starting from the initial edge $(i_s, i_f)$ and provides an interface with a single method $addBetween(a, o, b)$, which
  adds a new vertex $o$ with edges $(a, o)$ and $(o, b)$. CRDTs are strongly convergent by their definition, so if we can model a data type using CRDT, it is strongly convergent as well.
  
  Let us show that the operations of $G_{\vdash}$ may be expressed via $addBetween$, so $G_{\vdash}$ is a specific case of the state-based CRDT above.
  Initial state may be expressed as a set of edges $\{ (c_i, i_f) \}$, where $c_i$ are the register constructor operations,
  and $i_f$ is a sentinel.
  $add(\Gamma, o)$ may be expressed as executing $addBetween(o', o, i_f)$ for every premise $o' \in \Gamma$, and
  $rebase(o, \hat{o})$ is equivalent to performing $addBetween(o', \hat{o}, o)$ for every premise $o' \in \Gamma(o)$.

\end{proof}

\begin{theorem}
  Let all replicas freeze their histories $\mathcal{H}_{r_i}$ and have them {\itshape published}, so
  replica $r$ {\itshape synchronizes-with} all of them, resulting in local state $s$ produced by a local history $\mathcal{H}_r$ (which is induced from the entailment graph $G_{\vdash}$). 
  After all replicas {synchronize-with} published $\mathcal{H}_r$, they obtain the same state $s$ and their histories $\mathcal{H}_{r_i}'$ are inducible from $G_{\vdash}$.
  \[
    \forall i \colon \exists s \colon \forall r \in \mathcal{R} \colon s = \mathcal{I}_i(\mathcal{H}_r(M_i)).
  \]
\end{theorem}

\begin{proof}[Proof sketch]
  We need to prove that in absence of newly issued operations and centralized reconciliation of all replicas, which took place at $r$,
  all the replicas will acquire the same state during synchronization without calling {\bfseries\scshape resolve}, i.e., no further conflict resolution happens.

  Let history $\mathcal{H}'$ of a replica $r'$ be {\itshape induced}, i.e., generated as a topological sort, from $G_{\vdash}'$ before replica $r$ {\itshape synchronized-with} $r'$. In terms of graphs, upon
  $r'$ receives $\mathcal{H}_r$, $G_{\vdash}' \sqsubseteq G_{\vdash}$, as no new operations are made on $r$ (and, consequently, no new entailment edges are added).
  Then, as we previously showed that an entailment graph may be modeled as a state-based CRDT, due to the mononocity $G_{\vdash}' \leq G_{\vdash}$ and the idempotence of least-upper-bound operator $\sqcup$
  \cite{Shapiro2011SSSS}, $G_{\vdash}' \sqcup G_{\vdash} = G_{\vdash}$.

  The rest we need to prove is that any history induced from $G_{\vdash}$ produces the same state $s$
  which is modeled as a set of visible and invisible action instances.
  Due to discard-completeness of registers and absence of conflicts in $G_{\vdash}$, if $a_1$ is invisible, then $\exists! \ a_2 \colon a_1 \ll a_2$, i.e., $a_1 \prec a_2$.
  Thus, due to the happens-before order on actions that make other actions invisible in any history, all the histories produce the same state.
\end{proof}

\section{Examples}

In this section, we illustrate the proposed model with several register data type definitions by
specifying their actions $A$ and their visibility predicates $vis(a, \mathcal{H}_r)$.

\subsection{Register}

Register is a data type that stores a single value with an action basis $A = \{ mov \}$ and a query $val$,
which is instantiated with the constructor $c_x = [mov \ 0]$.
Let the state $s \in S$ be modeled as a pair $(\nu, \alpha)$,
where $\nu$ is the current value of the register and $\alpha$ is the last action that set this value.
Then an action $a = mov \ v$ is applied to local register replica $x_r$ with $\Gamma(a) = \{ \alpha(x_r) \}$ 
and interpreted by $\mathcal{I}$ as setting the state $s$ to $(v, a)$.
Getting the value is performed as follows: $val = \lambda . \nu$.
Visibility of an action $a$ is defined as the local state $s = \mathcal{I}(\mathcal{H}_r(x)) = (\nu, \alpha)$ of the register $x$ holds the value $v$
set by this action, i.e., $vis(a, \mathcal{H}_r(x)) \Leftrightarrow \nu = v \land \alpha = a$.
Any pair of concurrent operations containing $mov$ actions corresponding to the same register instance $x$ are in conflict.

\subsection{Arithmetic Register}

Let us extend an integer register data type with arithmetic actions $A = \{ mov, add, mul \}$
with an interpreter $\mathcal{I}$ producing the state $s = (\omega, \nu, \alpha)$, where
$\omega \in \{ =, +, \cdot \}$ is the current computation mode corresponding to the actions of $A$,
$\nu$ is the current initial value, and $\alpha$ is a sequence of actions performed on this value. 
Query $val$ may then be derived as getting the aggregate of the actions in $\alpha$ over $\nu$, i.e.,
$val = \lambda. \ foldr \ \omega \ \nu \ \alpha$.
Visibility of an action $a$ then may be defined as follows: $vis(a, \mathcal{H}_r(x)) \Leftrightarrow a \in \alpha$,
and $\Gamma(a) = \{ {last}(\alpha) \}$.

Let us consider an example with a register $x$ constructed by $c_i = [a_0]$, $a_0 = add \ 0$, which is interpreted as the state
$s = (+, 0, \{a_0\})$.
A following addition action $a_1 = add \ 2$ is interpreted as appending $a_1$ to $\alpha$, i.e., $s = (+, 0, \{a_0, a_1\})$.
Similarly, a multiplication action $a_2 = mul \ 5$ is interpreted as setting
the computation mode $\cdot$ with the current $\nu = val$ and overwriting the action sequence with $\{ a_2 \}$, i.e., $s = (\cdot, 2, \{a_2\})$.
An assignment action $a_3 = mov \ 3$ modifies the state to $s = (=, 10, \{ a_3 \})$, but subsequent assignments are not accumulated in the action set.
A similar example is shown in Figure \ref{fig:arithmetic-register-conflict}.

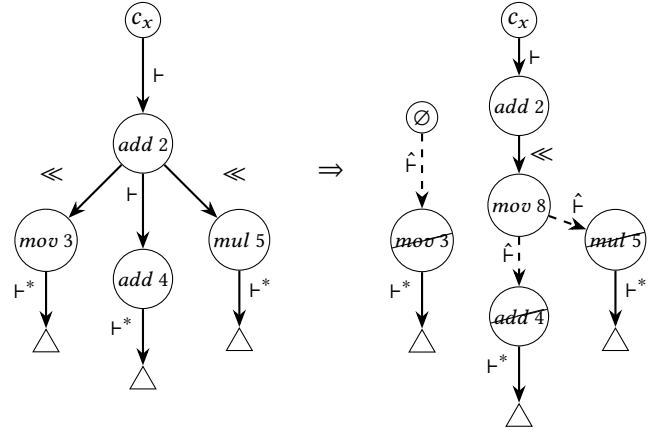
\begin{figure}
  \centering
  \begin{tikzpicture}[node distance=1.0cm]
    \node[obj] (cx) {$c_x$};
      
      \node[obj, below=of cx] (add2) {\footnotesize $add\ 2$};
      \draw[entails] (cx) -- node[right] {$\vdash$} (add2);
      
      \node[obj, below left=of add2] (mov3) {\footnotesize $mov\ 3$};
      \node[obj, below=of add2] (add4) {\footnotesize $add\ 4$};
      \node[obj, below right=of add2] (mul5) {\footnotesize $mul\ 5$};
      
      \draw[entails] (add2) -- node[above left, xshift=-8pt] {$\ll$} (mov3);
      \draw[entails] (add2) -- node[above, xshift=-4pt] {$\vdash$} (add4);
      \draw[entails] (add2) -- node[above right, xshift=8pt] {$\ll$} (mul5);
      
      \node[base, below=0.75cm of mov3] (b1) {};
      \node[base, below=0.75cm of add4] (b2) {};
      \node[base, below=0.75cm of mul5] (b3) {};
      
      \draw[entails] (mov3) -- node[above, xshift=-8pt, yshift=-3pt] {$\vdash^*$} (b1);
      \draw[entails] (add4) -- node[above, xshift=-8pt, yshift=-3pt] {$\vdash^*$} (b2);
      \draw[entails] (mul5) -- node[above, xshift=8pt, yshift=-3pt] {$\vdash^*$} (b3);

      \node at (2.5, -2) {$\Rightarrow$};

      \node[obj] (cx2) at (5.0, 0.0) {$c_x$};
      
      \node[obj, below=of cx2, yshift=0.5cm] (add2b) {\footnotesize $add\ 2$};
      \draw[entails] (cx2) -- node[right] {$\vdash$} (add2b);

        \node[obj, below left=of add2b, yshift=-0.5cm] (mov3b) {\footnotesize $\cancel{mov\ 3}$};
        \node[obj, below=of add2b, yshift=-1.0cm] (add4b) {\footnotesize $\cancel{add\ 4}$};
        \node[obj, below right=of add2b, yshift=-0.5cm] (mul5b) {\footnotesize $\cancel{mul\ 5}$};   
      
      \node[obj, below=0.5cm of add2b] (mov8) {\footnotesize $mov\ 8$};

      \node[obj, above=1.0cm of mov3b] (empty) {$\varnothing$};   \draw[entails, dashed] (empty) -- node[above left, xshift=2pt, yshift=-4pt] {$\hat{\vdash}$} (mov3b);

      \draw[entails, dashed] (mov8) -- node[above, xshift=-4pt, yshift=-4pt] {$\hat{\vdash}$} (add4b);
      \draw[entails, dashed] (mov8) -- node[above right, xshift=-2pt] {$\hat{\vdash}$} (mul5b);
      \draw[discards] (add2b) -- node[right, yshift=0pt] {$\ll$} (mov8);
      
      \node[base, below=0.75cm of mov3b] (b1b) {};
      \node[base, below=0.75cm of add4b] (b2b) {};
      \node[base, below=0.75cm of mul5b] (b3b) {};
      
      \draw[entails] (mov3b) -- node[above, xshift=-8pt, yshift=-3pt] {$\vdash^*$} (b1b);
      \draw[entails] (add4b) -- node[above, xshift=-8pt, yshift=-3pt] {$\vdash^*$} (b2b);
      \draw[entails] (mul5b) -- node[above, xshift=8pt, yshift=-3pt] {$\vdash^*$} (b3b);

  \end{tikzpicture}
  \caption{Conflict resolution on arithmetic register within additive context of $add \ 2$, which resulted in $mov \ 8$.}
  \Description{An arithmetic register is constructed with $c_x$ entailing $add \ 2$.
  Then, three concurrent operations $mov \ 3$, $add \ 4$, and $mul \ 5$ are made, where
  $mov \ 3$ and $mul \ 5$ discard $add \ 2$, while $add \ 4$ entails it.
  After the conflict is resolved, a new operation $mov \ 8$ is made, which discards $add \ 2$.
  Specifically, $mov \ 3$ is cancelled, while $add \ 4$ and $mul \ 5$ are rebased to be entailed by $mov \ 8$.}
  \label{fig:arithmetic-register-conflict}
\end{figure}

\subsection{LWW-Register}

Last-Write-Wins Register is a conflict-free replicated data type \cite{Shapiro2011}
that prioritizes the latest value assigned to the register in case of concurrent writes.
To achieve this, each write operation $mov \ x \ t$ includes a total order timestamp $t \in \mathcal{T}$,
usually derived from the wall-clock time of the issuing process.
A query $val = \{v(a) \mid t(a) = \max_{t}(\mathcal{H}_{r}(x))\}$ simply returns the values associated with the greatest timestamp
which follows the celebrated Thomas write rule \cite{Thomas1979}.

Despite LWW-Register being seen as conflict-free, still the conflict can happen, which is
known as {\itshape lost update problem}, when
a write operation effect may have never been observed by a read operation
due to the presence of a more recent concurrent write operation, i.e., 
values are silently overwritten during synchronization.
Traditionally, this nuance is ignored in favor of simplicity of the data type interface,
however, it is easy to resolve such situations in entails-discards model.

To identify lost updates, we specify visibility predicate as 
$vis(a, \mathcal{H}_{r}(x)) \Leftrightarrow \{a\} = \{ a' \mid t(a') = \max_{t}(\mathcal{H}_{r}(x)) \}$,
query $val = \{ v(a) \mid t(a) = \max_{t}(\mathcal{H}_{r}(x)) \}$,
and $\Gamma(a) = \{ a \mid t(a) = \max_{t}(\mathcal{H}_{r}(x)) \}$.
Conflicts may be resolved automatically by the replicas in place, when
the merge operation $\hat{o}$ simply replays the latest seen write operation. At the same time,
equal greatest timestamps may be still defined as a conflict, as shown in Figure \ref{fig:lww-register-conflict}, or resolved deterministically,
e.g., by comparing operation identifiers.

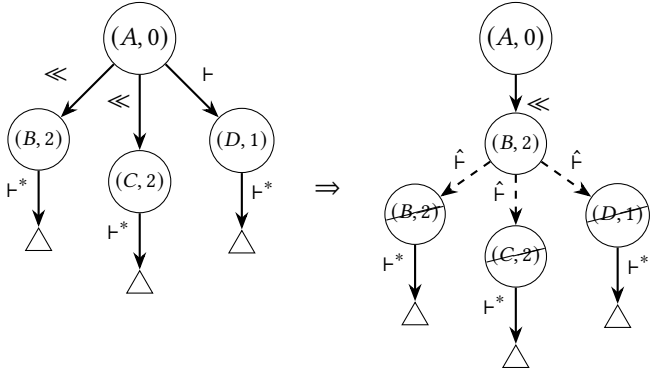
\begin{figure}
  \centering
  \begin{tikzpicture}[node distance=1.0cm]
    \node[obj] (cx) {$(A, 0)$};
    
    \node[obj, below left=of cx] (o1) {\footnotesize $(B, 2)$};
    \node[obj, below=of cx] (o3) {\footnotesize $(C, 2)$};
    \node[obj, below right=of cx] (o2) {\footnotesize $(D, 1)$};
    
    \draw[entails] (cx) -- node[above, xshift=-8pt, yshift=-2pt] {$\ll$} (o3);
    \draw[entails] (cx) -- node[above right, xshift=0pt] {$\vdash$} (o2);
    
    \draw[discards] (cx) -- node[above, xshift=-12pt] {$\ll$} (o1);
    
    \node[base, below=0.75cm of o1] (b1) {};
    \node[base, below=0.75cm of o3] (b3) {};
    \node[base, below=0.75cm of o2] (b2) {};
    
    \draw[entails] (o1) -- node[above, xshift=-8pt, yshift=-3pt] {$\vdash^*$} (b1);
    \draw[entails] (o3) -- node[above, xshift=-8pt, yshift=-3pt] {$\vdash^*$} (b3);
    \draw[entails] (o2) -- node[above, xshift=8pt, yshift=-3pt] {$\vdash^*$} (b2);

    \node at (2.5, -2) {$\Rightarrow$};

    \node[obj] (cxb) at (5.0, 0.0) {$(A, 0)$};
    
    \node[obj, below left=of cxb, yshift=-1.0cm] (o1b) {\footnotesize $\cancel{(B, 2)}$};
    \node[obj, below=of cxb, yshift=-1.0cm] (o3b) {\footnotesize $\cancel{(C, 2)}$};
    \node[obj, below right=of cxb, yshift=-1.0cm] (o2b) {\footnotesize $\cancel{(D, 1)}$};
    
    \node[obj, below=0.5cm of cxb] (fix) {\footnotesize $(B, 2)$};

    \draw[entails, dashed] (fix) -- node[above left, yshift=-4pt] {$\hat{\vdash}$} (o3b);
    \draw[entails, dashed] (fix) -- node[above right, xshift=-2pt] {$\hat{\vdash}$} (o2b);
    \draw[entails, dashed] (fix) -- node[above right, xshift=-8pt] {$\hat{\vdash}$} (o1b);
    \draw[discards] (cxb) -- node[right, yshift=-4pt] {$\ll$} (fix);
    
    \node[base, below=0.75cm of o1b] (b1b) {};
    \node[base, below=0.75cm of o3b] (b3b) {};
    \node[base, below=0.75cm of o2b] (b2b) {};
    
    \draw[entails] (o1b) -- node[above, xshift=-8pt, yshift=-3pt] {$\vdash^*$} (b1b);
    \draw[entails] (o3b) -- node[above, xshift=-8pt, yshift=-3pt] {$\vdash^*$} (b3b);
    \draw[entails] (o2b) -- node[above, xshift=8pt, yshift=-3pt] {$\vdash^*$} (b2b);

  \end{tikzpicture}
  \caption{Conflicts in LWW-Register: $(B, 2)$ and $(C, 2)$ are concurrent operations with equal timestamps; $(D, 1)$ is suppressed by its siblings.}
  \Description{TODO: add description}
  \label{fig:lww-register-conflict}
\end{figure}

\subsection{Multi-Register}

Finally, collaboration in shared memory consisting of multiple registers can be considered, which illustrates that
operations may contain actions on distinct registers.

Let us introduce a $touch$ action to allow definition of custom user-defined premises for operations
that span multiple registers. Action $touch$ simply forces recent $mov$ operation $o$ to be included in the premises
$\Gamma$ of the issuing operation.
Roughly speaking, this mechanism enables expression of custom optimistic locks on registers by 
referencing operations providing the currently observed values.

Let us consider an example with three registers corresponding to fields of a calendar event: 
$M_1$ is a string event title, $M_2$ is a pair of the event start and the end time modeled as unix timestamps, and
$M_3$ is the event location.
All the registers are initially instantiated by their respective constructors: $c_1 = mov \ '{\text{Lunch: Alice x Bob}}'$,
$c_2 = mov \ {\text{1pm-1.30pm}}$, $c_3 = mov \ '{\text{Bambi's}}'$.

Assume Alice has acquired a time slot to extend the event duration until 2pm,
but Bob changes the event location concurrently,
i.e., $o_A = [mov \ M_2 \ {\text{1pm-2pm}}]$ exists with $o_B = [mov \ M_3 \ {\text{Meadow's}}]$.
However, Alice had optimistically assumed that the event location will not be changed by others,
because, otherwise, the event duration would not have been extended in $o_A$.

These operations are compatible, as they operate on different registers,
but they can be simply adjusted, so that the intention conflict and 
optimistic assumption of Alice about the event location not changed is respected.
To achieve this, Alice's operation $o_A$ includes a $touch \ M_3$ action in $o_A$,
so that Bob's concurrent update $o_B$ will result in a conflict, as soon as Alice receives Bob's update,
as shown in Figure \ref{fig:multi-register}.

\begin{figure}
  \centering
  \begin{tikzpicture}[node distance=1.0cm]
    \node[obj] (c1) {$c_1$};
    \node[obj, right=2.5cm of c1] (c2) {$c_2$};
    \node[obj, right=2.5cm of c2] (c3) {$c_3$};
    
    \node[obj, below left=0.75cm and 0.75cm of c2] (oA) {$o_A$};
    \draw[entails] (c2) -- node[above, xshift=-8pt] {$\ll$} (oA);
    \draw[entails] (c3) -- node[above, xshift=0pt, yshift=0pt] {$\vdash$} (oA);
    
    \node[obj, below right=0.75cm and 0.75cm of c2] (oB) {$o_B$};
    \draw[entails] (c3) -- node[above, yshift=-12pt, xshift=8pt] {$\ll$} (oB);
    
    \node[base, below=0.75cm of oA] (bA) {};
    \node[base, below=0.75cm of oB] (bB) {};
    \draw[entails] (oA) -- node[above, xshift=-8pt, yshift=-3pt] {$\vdash^*$} (bA);
    \draw[entails] (oB) -- node[above, xshift=8pt, yshift=-3pt] {$\vdash^*$} (bB);
  \end{tikzpicture}
  \caption{Multi-register example: Alice issues $o_A = [touch \ M_3, \ mov \ M_2 \ {\text{1pm-2pm}}]$, Bob issues $o_B = [mov \ M_3 \ {\text{Meadow's}}]$.
  Conflict arises as the effect of $c_3$ is concurrently discarded by $o_B$.}
  \label{fig:multi-register}
\end{figure}
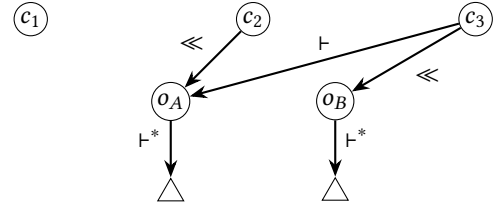

\section{Related work}

The presented conflict model is influenced by dynamic epistemic logic \cite{DynamicEpistemicLogic2008}
with a goal to revisit {\itshape intention preservation} property \cite{Sun1998} firstly captured
in consistency model for Operational Transformations \cite{EllisGibbs1989}.
As operation may be applied only when its premises are already present in the history,
the conflict model implies a weaker form of causal delivery \cite{Lamport1978}.
Support of multi-register data types accords with wait-free data structures in shared memory \cite{AspnesHerlihy1990}
and the composition of convergent and commutative replicated data types \cite{Baquero2017,Weidner2020}.

The entailment graph is a monotonic semi-lattice, i.e.,
a {\itshape state-based replicated data type} (CvRDT) \cite{Shapiro2011} and, noteworthy,
it is used to represent the {\itshape operation-based} history of a data type.
Mergeable replicated data types (MRDTs) \cite{Kaki2019} also present a framework for
defining arbitrary data types using invertible relational specifications,
where {\itshape concretization} functions are employed to resolve conflicts by ordering operations
in a three-way merge manner \cite{Mens2002}. However, conflicts are not explicitly represented neither in CRDTs nor in MRDTs,
which complicates reasoning about the data type behavior in presence of concurrent updates in collaboration.

\section{Conclusion}

In this work, we presented a conflict model for collaborative data types
based on custom defined entails-discards semantics of operations.
The next step in this research is to refine the correctness proof and 
explore the expressiveness of the model by revisiting known non-register CRDTs, to define a general procedure for
selective undo-redo operations, and implement a software framework to define custom data types 
and automatically prove their correctness properties.

\bibliographystyle{ACM-Reference-Format}
\bibliography{sample-base}

\end{document}